\documentclass[conference,10pt]{IEEEtran}
\IEEEoverridecommandlockouts

\makeatletter
\def\ps@headings{%
\def\@oddhead{\mbox{}\scriptsize\rightmark \hfil \thepage}%
\def\@evenhead{\scriptsize\thepage \hfil \leftmark\mbox{}}%
\def\@oddfoot{}%
\def\@evenfoot{}}
\makeatother 

\usepackage{latexsym}
\usepackage{epsfig}
\usepackage{amsmath}
\usepackage{amsfonts}
\usepackage{amssymb}
\usepackage{wrapfig}
\usepackage{float}
\usepackage{bm}

\newtheorem{theorem}{Theorem}
\newtheorem{definition}{Definition}
\newtheorem{lemma}{Lemma}
\newtheorem{proposition}{Proposition}
\newtheorem{corollary}{Corollary}
\newtheorem{remark}{Remark}

\newenvironment{proof}[1][Proof]{\begin{trivlist}
\item[\hskip \labelsep {\bfseries #1}]}{\end{trivlist}}

\begin{document}

\title{Fast Mixing of Parallel Glauber Dynamics and Low-Delay CSMA Scheduling\thanks{Authors names appear in the alphabetical order
of their last names.}\thanks{This work is supported by MURI grant BAA 07-036.18, NSF Grants 07-21286,
05-19691, 03-25673, AFOSR Grant FA-9550-08-1-0432 and DTRA Grant HDTRA1-08-1-0016.}}

\author{\IEEEauthorblockN{Libin Jiang}
\IEEEauthorblockA{UC Berkeley\\
ljiang@eecs.berkeley.edu} \and \IEEEauthorblockN{Mathieu Leconte}
\IEEEauthorblockA{Technicolor Paris Lab\\
mathieu.leconte@technicolor.com} \and \IEEEauthorblockN{Jian Ni}
\IEEEauthorblockA{UIUC\\
jianni@illinois.edu} \and \IEEEauthorblockN{R. Srikant}
\IEEEauthorblockA{UIUC\\
rsrikant@illinois.edu} \and \IEEEauthorblockN{Jean Walrand}
\IEEEauthorblockA{UC Berkeley\\
wlr@eecs.berkeley.edu}}

\maketitle

\begin{abstract}

Glauber dynamics is a powerful tool to generate randomized, approximate solutions to combinatorially
difficult problems. It has been used to analyze and design distributed CSMA (Carrier Sense Multiple Access)
scheduling algorithms for multi-hop wireless networks. In this paper we derive bounds on the mixing time of a
generalization of Glauber dynamics where multiple links are allowed to update their states in parallel and
the fugacity of each link can be different. The results can be used to prove that the average queue length
(and hence, the delay) under the parallel Glauber dynamics based CSMA grows polynomially in the number of
links for wireless networks with bounded-degree interference graphs when the arrival rate lies in a fraction
of the capacity region. We also show that in specific network topologies, the low-delay capacity region can
be further improved.

\end{abstract}

\section{Introduction}

In wireless networks, the links (transmitter-receiver pairs) may not be able to transmit simultaneously due
to interference. A scheduling algorithm, or MAC (Medium Access Control) protocol, determines which links can
access the medium in each time instant so that no active links interfere with each other. Since many wireless
network applications today have stringent bandwidth and delay requirements but the resources (e.g., spectrum,
power) are often quite limited in a wireless setting, designing low-complexity scheduling algorithms to
achieve high throughput and low delay is of great importance.

It is well known that the queue-length based \emph{Maximum Weighted Scheduling} (MWS) algorithm is
\emph{throughput-optimal} \cite{taseph92}, meaning that it can stabilize the network queues for all arrival
rates in the capacity region of the network. However, MWS requires the network to select a max-weight
independent set in the interference graph in every time slot, which is NP-hard for general interference
graphs \cite{garjoh79}. There exist several low-complexity alternatives such as \emph{Maximal Scheduling} and
\emph{Greedy Maximal Scheduling}, but in general these algorithms can only guarantee to achieve a fraction of
the capacity region (see \cite{lecnisri09} and references therein).

Due to their simplicity, random access type scheduling algorithms such as Aloha and Carrier Sense Multiple
Access (CSMA) are widely used in practice. Performance analysis of random access algorithms in single-hop
wireless networks can be found in \cite{bergal92}. In \cite{boo87} the authors introduced a continuous-time
Markov chain model to analyze the performance of a (fixed-parameter) CSMA algorithm in multi-hop wireless
networks, and it was shown that the stationary distribution over the schedules has a product form. The model
was used in \cite{wankar05} to study throughput and fairness issues of the CSMA algorithm, and its
insensitivity properties were studied in~\cite{liekaileuwon07}. A discrete-time version of the algorithm was
studied in \cite{nisri09,nitansri10}.

In \cite{jiawal08} the authors proposed an adaptive CSMA algorithm where the links adaptively adjust their
parameters based on locally measured arrival and service rates. The algorithm was shown to be
throughput-optimal under a time-scale separation assumption (the CSMA Markov chain converges to its
steady-state distribution instantaneously compared to the time-scale of adaptation of the CSMA parameters)
which can be justified using a stochastic-approximation argument \cite{liuyipro08,CDC_convergence}. In
\cite{rajshashi09} the authors established throughput-optimality of their adaptive CSMA algorithm without the
time-scale separation assumption by choosing the link parameters to be slowly varying functions of the queue
lengths. The discrete-time equivalent of this analysis for the model in \cite{nisri09,nitansri10} appears in
\cite{ghasri10}.

Central to these CSMA algorithms is the so-called \emph{Glauber dynamics}, which is a Markov Chain Monte
Carlo method that can be used to sample the independent sets of a graph according to a product-form
distribution \cite{dyegre00,vig01}. Under traditional Glauber dynamics, in each time slot one link will be
selected uniformly at random, and only that link can change its state while other links will keep their
states unchanged. For the chosen link, if all of its neighboring links were in state $0$ (inactive) in the
previous slot, then the link will choose to be in state $1$ (active) with probability
$\frac{\lambda}{1+\lambda}$ and in state $0$ with probability $\frac{1}{1+\lambda}$; otherwise (i.e., if at
least one of its neighboring links was in state $1$ in the previous slot), the link will choose to be in
state $0$ definitely. In statistical physics, the parameter $\lambda$ is called the \emph{fugacity} since it
indicates how likely a selected site will change its state.\footnote{And the same fugacity $\lambda$ is
shared by all links under the traditional Glauber dynamics. In the current paper, however, we need to
consider heterogeneous fugacities. That is, different links have different fugacities.} Glauber dynamics has
many applications in statistical physics, graph coloring, approximate counting, and combinatorial
optimization (e.g., \cite{dyegre00,mar99,vig01}). In these applications, the performance of the Glauber
dynamics is often determined by how fast the Markov chain converges to the stationary distribution. The
Glauber dynamics is said to be \emph{fast mixing} if the \emph{mixing time} (will be defined formally later)
grows polynomially in the size of the graph. The algorithm in \cite{nisri09,nitansri10} is a generalization
of Glauber dynamics where multiple links are allowed to update their states in parallel. The CSMA algorithm
in \cite{boo87} can be viewed as continuous-time Glauber dynamics with fixed parameters, and the CSMA
algorithms in \cite{jiawal08,rajshashi09} can be viewed as continuous-time Glauber dynamics with adaptive
parameters. In this paper, we focus on the discrete-time CSMA algorithm suggested in
\cite{nisri09,nitansri10}. An important feature of this algorithm is that the \emph{overhead} for signaling
is constant (independent of the size of the network) even taking into account the collisions during the
signaling phase. The results in this paper can be extended to other versions of CSMA algorithms as well.

The recently proposed Glauber dynamics based CSMA algorithms have made an important progress to the design of
low-complexity distributed scheduling algorithms to achieve maximum throughput in wireless networks. On the
other hand, the delay performance of these CSMA algorithms has not been well understood. As shown in a recent
work \cite{shatsetsi09}, for general networks, it may not be possible to design low-complexity scheduling
algorithms which can achieve both low delay (i.e., grows polynomially in the number of links) and even a
diminishingly small fraction (i.e., approaches zero when the number of links increases) of the capacity
region unless $\mathbf{NP}\subseteq \mathbf{BPP}$ or $\mathbf{P}=\mathbf{NP}$. In this paper we prove a
positive result. We show that our parallel Glauber dynamics based CSMA scheduling algorithm can achieve both
low delay and a fraction (independent of the size of the network) of the capacity region  when the
interference graphs satisfy certain properties.

The main contributions of this paper include:
\begin{itemize}
\item We analyze the mixing time of Glauber dynamics with parallel updates and heterogenous fugacities. We
derive various conditions on the system parameters such as fugacities, vertex degrees and update
probabilities, under which the mixing time grows logarithmically in the number of vertices.

\item Based on the above mixing time results, we show that, for wireless networks with bounded-degree
interference graphs, the parallel Glauber dynamics based CSMA algorithm can achieve a small queue length
($O(\log n)$ where $n$ is the number of links in the network) at each link if the arrival rate lies within a
fraction (independent of $n$) of the capacity region and the fugacities of the links are appropriately chosen
and fixed. Moreover, we consider an adaptive version of the CSMA algorithm where the fugacities of the links
are adjusted based on local queue length information. We show that the total queue length in the network
grows polynomially ($O(n^3\log n)$) under the adaptive CSMA algorithm.

\item Unlike prior analysis of Glauber dynamics based CSMA algorithms which uses the  \emph{conductance}
method to obtain exponential bounds on the mixing time and the queue lengths, here we use the \emph{coupling}
method which allows us to obtain polynomial bounds on the queue lengths for a fraction of the capacity
region.

\item For a special but important network topology (Wireless LANs), we show that Glauber dynamics based CSMA
can support the full capacity region with linear mixing time.

\end{itemize}

In a related work \cite{shashi10}, the authors proposed a CSMA algorithm which can achieve order-optimal
delay performance (i.e., the per-node delay is bounded by a constant) for networks with \emph{polynomial
growth structure} (the number of $r$-hop neighbors of any node is bounded by a polynomial of $r$). The key
idea is to (periodically) partition the network into sub-networks with bounded number of nodes in each
sub-network such that the mixing time of the Markov chain of the schedule in each sub-network is bounded by a
constant. However, the partition is realized by a distributed coloring algorithm which requires the nodes to
exchange messages with their multi-hop neighbors, and this overhead can be significant and may increase with
the size of the network, especially when the discrete nature of signaling is considered. So a direct
comparison of our results and the results in \cite{shashi10} appears to be difficult.

The paper is organized as follows. In Section~\ref{sec:parallelGD} we introduce a CSMA scheduling algorithm
based on a generalization of Glauber dynamics with parallel updates and heterogenous fugacities. In
Section~\ref{sec:mixingPGD} we derive bounds on the mixing time of parallel Glauber dynamics, with the proof
presented in Section~\ref{s.proof-fast}. In Section~\ref{sec:queuelength} we analyze the delay performance of
the CSMA scheduling algorithms using the mixing time results. Section~\ref{s.complete-graph} is dedicated to
the analysis of complete interference graphs. The paper is concluded in Section~\ref{sec:conclusion}.

\section{CSMA Scheduling Based on Parallel Glauber dynamics}  \label{sec:parallelGD}

For a wireless link $i$, we use $\mathcal{N}_i$ to denote the set of conflicting links (called
\emph{conflict} set or \emph{neighbor} set) of link $i$: if any link in $\mathcal{N}_i$ is active
(transmitting), then link $i$ cannot be active.

The interference relationship among the wireless links can be represented by the so-called \emph{interference
graph} (or \emph{conflict graph}) $G=(V,E)$, where the vertices in $V$ represent \emph{wireless links} in the
network, and there is an edge between two vertices in $G$ if the corresponding wireless links interfere with
each other. (For example, Fig.\ref{fig:3-link}(a) shows a wireless network of 3 links, where link $2$
interferes with links $1$ and $3$, and links $1$ and $3$ don't interfere with each other. So, the associated
interference graph is shown in Fig.\ref{fig:3-link}(b), and the conflict set of link 2 is
$\mathcal{N}_2=\{1,3\}$, etc.) An \emph{independent set} of $G$ is a subset of the vertices in $V$ where no
two vertices are neighbors of each other. Let $\mathcal{I}$ be the set of all independent sets of $ G$.

\begin{figure}
\includegraphics[width=8cm]{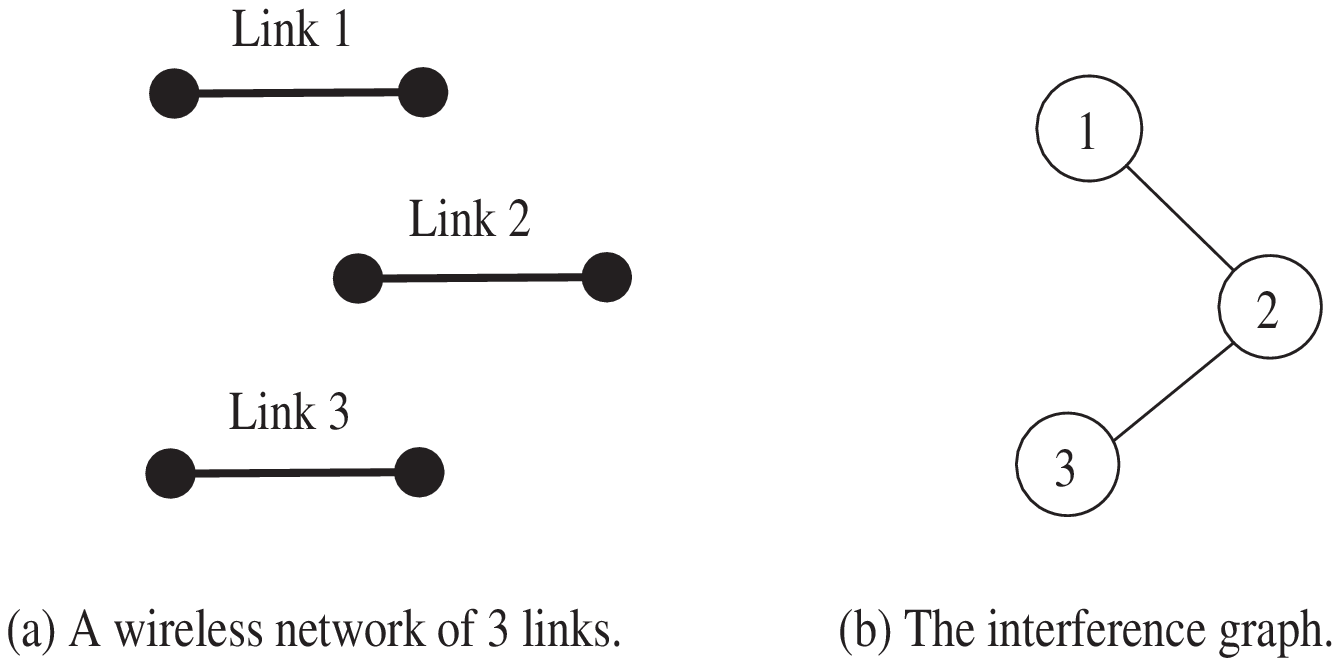}
\caption{An example network and the corresponding interference graph.} \label{fig:3-link}
\end{figure}

A \emph{feasible} \emph{schedule} of the network is a set of wireless links that can be active at the same
time according to the conflict set constraint, i.e., no two links in a feasible schedule conflict with each
other. This corresponds to an independent set of vertices in the interference graph. (In
Fig.\ref{fig:3-link}(b), for example, the sets $\{1\},\{1,3\},\{2\}$ are all feasible schedules.) We will use
\emph{links} and \emph{vertices} interchangeably throughout this paper.

Suppose $|V|=n$. We can represent a feasible schedule by a vector $\sigma$ of the form $(\sigma_i)_{i\in V}$,
with $\sigma_i\in \mathcal{X}=\{0, 1\}$ for all $i\in V$. For a link $i$ and a schedule $\sigma$, we say link
$i$ is included in the schedule (written as $i\in \sigma$) if $\sigma_i=1$. Note that $\sigma$ is a feasible
schedule if the set $\{i\in V:\sigma_i=1\}$ is an independent set of $G$, i.e., if $\sigma_i+\sigma_j \leq 1,
\mbox{ for all } (i,j)\in  E.$ Let $\Omega\subseteq \mathcal{X}^n$ be the set of all feasible schedules on
$G$.

We consider a time-slotted system. A \emph{scheduling algorithm} is a procedure to decide which schedule to
be used in every time slot for data transmission.  In \cite{nisri09} we proposed a scheduling algorithm based
on a generalization of Glauber dynamics (called \emph{parallel Glauber dynamics}) where multiple links are
allowed to update their states in a single time slot. The key idea is that in every time slot, we select an
independent set of links $\mathbf m \in \mathcal I$ to update their states according to a randomized
procedure, i.e., we select $\mathbf m \in \mathcal I$ with probability $q_{\mathbf m}$, where $\sum_{\mathbf
m \in \mathcal I} q_{\mathbf m}=1.$ We call $\mathbf m$ the \emph{decision schedule}.\footnote{A distributed
mechanism to generate the decision schedule, as suggested in \cite{nisri09}, is the following. At the
beginning of a slot, each link independently transmits a short INTENT message with probability $a\in (0,1)$.
A link is included in the decision schedule if (and only if) it sends an INTENT message while none of its
neighbors sends such a message.}

The parallel Glauber dynamics is formally described as follows.

\noindent\hrulefill
\\
\textbf{Parallel Glauber Dynamics (in Time Slot $t$)}

\noindent\hrulefill
\begin{itemize}
\item[1.] Randomly choose a decision schedule $\mathbf m(t)\in \mathcal I$ with probability $q_{\mathbf
m(t)}$.

\item[2.] For every link $i\in \mathbf m(t)$:\\
\phantom{aa} \textbf{If} $\sum_{j\in \mathcal{N}_i} \sigma_j(t-1)=0$  \\
\phantom{aaaa} (a) $\sigma_i(t)=1$ with probability $p_i=\frac{\lambda_i}{1+\lambda_i}.$ \\
\phantom{aaaa} (b) $\sigma_i(t)=0$ with probability $\bar{p}_i=\frac{1}{1+\lambda_i}.$ \\
\phantom{aa} \textbf{Else}\\
\phantom{aaaa} (c) $\sigma_i(t)=0$.

\item[] For every link $j \notin \mathbf m(t):$\\
\phantom{aaaa} (d) $\sigma_j(t)=\sigma_j(t-1)$.

\end{itemize}
\noindent\hrulefill \medskip{}

Under the Parallel Glauber Dynamics based CSMA (called \textbf{PGD-CSMA} for short), $\sigma(t)$ is used as
the \emph{transmission schedule} in time slot $t$: link $i$ will transmit a data packet if $\sigma_i(t)=1$,
and will keep silent if $\sigma_i(t)=0$. Note that link $i$ knows whether $\sum_{j\in \mathcal{N}_i}
\sigma_j(t-1)=0$ by conducting \emph{carrier sensing} in time slot $t-1$: the channel (medium) will be sensed
idle if none of its neighboring links were transmitting (i.e., $\sum_{j\in \mathcal{N}_i} \sigma_j(t-1)=0$).
$p_i$ is called the \emph{link activation probability}, which is determined by the fugacity $\lambda_i$ of
link $i$.

We can show that if the transmission schedule used in the previous slot and the decision schedule used in the
current slot both are feasible, then the transmission schedule generated in the current slot is also feasible
\cite{nisri09}. Moreover, given the fugacities $\lambda_i$'s, because $\sigma(t)$ only depends on the
previous schedule $\sigma(t-1)$ and some randomly selected decision schedule $\mathbf m(t)$, $\sigma(t)$
evolves as a discrete-time Markov chain (DTMC).

\begin{theorem} \label{theorem:productform}
\emph{(\cite{nisri09}) A necessary and sufficient condition for the parallel Glauber dynamics to be
irreducible and aperiodic is $\cup_{\mathbf m\in \mathcal I: q_{\mathbf m}>0}\mathbf m = V,$ or equivalently,
the probability of selecting link $i$ in the decision schedule $q_i := \sum_{\mathbf m\ni i}q_{\mathbf m}>0$
for all $i\in V$. In this case the Markov chain is reversible and has the following product-form stationary
distribution:}
\begin{eqnarray}
\pi(\sigma) & = & \frac{\prod_{i \in \sigma}\lambda_i}{\sum_{\sigma' \in \Omega} \prod_{i \in
\sigma'}\lambda_i}. \label{eq:productformPGD}
\end{eqnarray}
\end{theorem}


Based on the product-form distribution, one can establish throughput-optimality of PGD-CSMA by either
choosing the link activation probabilities (fugacities) as appropriate increasing functions of the
(time-varying) queue lengths, or adjusting the fugacities based on the measured arrival and service rates, as
in \cite{jiawal08,nisri09,rajshashi09}. The focus of this paper is to analyze the mixing time of parallel
Glauber dynamics and the delay performance of PGD-CSMA. We will show that the parallel Glauber dynamics is
fast mixing when the fugacities satisfy certain conditions, and this implies that PGD-CSMA induces small
queue lengths when the arrival rates lie in a fraction of the capacity region.

\section{Mixing Time of Parallel Glauber Dynamics}  \label{sec:mixingPGD}

\subsection{Definitions} \label{sec:prelim}

Consider a finite-state, irreducible, aperiodic Markov chain $(P,\Omega, \pi)$ where $P$ denotes the
transition matrix, $\Omega$ denotes the state space, and $\pi$ denotes the unique stationary distribution.
First we describe a notion of distance between distributions.
\begin{definition}
\emph{The \emph{variation distance} between two distributions $\mu, \mu'$ on $\Omega$ is defined as}
\begin{eqnarray}
||\mu-\mu'||_{var} = \frac{1}{2}\sum_{x\in \Omega}|\mu(x)-\mu'(x)|.
\end{eqnarray}
\end{definition}
Note that $0\leq ||\mu-\mu'||_{var} \leq 1$ and $||\mu-\mu'||_{var}=0$ if and only if $\mu=\mu'$.

\begin{definition}
\emph{The \emph{mixing time} $T_{mix}$ of the Markov chain is defined as the time required for the Markov
chain to get close to the stationary distribution. More precisely,}
\begin{eqnarray}
T_{mix} = \max_{x\in \Omega} \inf\Big\{t: ||\mu_{x,t}-\pi||_{var} \leq \frac{1}{e}\Big\}
\end{eqnarray}
\emph{where $\mu_{x,t}$ is the distribution of the Markov chain at time $t$ if the Markov chain starts with
state $x$.}
\end{definition}

\subsection{Conditions for fast mixing of Parallel Glauber Dynamics}

We now state the main theorem of this section and a corollary on the mixing time of parallel Glauber dynamics
with fixed fugacities $\boldsymbol \lambda$. The proof of the main theorem will be presented in the next
section.

We associate a weight $f(v)>0$ to each vertex (or link) $v \in V$, and call $f(\cdot)$ the \emph{weight
function}.

\begin{theorem} \label{thm:main}
\emph{For any positive weight function $f(v)$ of $v\in V$, let $m=\min_{v\in V}f(v)$, $M=\max_{v\in V}f(v)$,
and $\xi=\frac{M}{m}$. If}
\begin{eqnarray} \label{eq:theta}
\theta \triangleq \min_{v\in V}\left\{q_v f(v)-\sum_{w\in\mathcal N_v}q_w\frac{\lambda_w
}{1+\lambda_w}f(w)\right\} > 0,
\end{eqnarray}
\emph{where $q_v$, as defined before, is the probability that vertex $v$ is included in the decision
schedule, then under the parallel Glauber dynamics,}
\begin{equation} \label{eq:dvboundPGD}
||\mu_{x,t}-\pi||_{var}  \leq  \min\Big\{1, (1-\frac{\theta}{M})^{t} n\xi\Big\},\forall x \in \Omega.
\end{equation}
\emph{Therefore the parallel Glauber dynamics is fast mixing and its mixing time is bounded by:}
\begin{eqnarray} \label{eq:mixingPGB}
T_{\textit{mix}} & \leq &  \bar{T}_{mix}  = \Big \lceil \frac{M}{\theta}\log(n\xi e) \Big\rceil.
\end{eqnarray}
\end{theorem}

One can specify the weight function $f(\cdot)$ to obtain various conditions on the fugacities for fast
mixing. We will show one of them below which will be used later to analyze the delay performance of PGD-CSMA.
Other conditions are given in Appendix~\ref{conditions}.

\begin{corollary} \label{coro:c3}
Let $m=\min_{v\in V}\frac{d_v}{q_v}$, $M=\max_{v\in V}\frac{d_v}{q_v}$ and $\xi=\frac{M}{m}$, where $d_v$ is
the degree of $v$ in the interference graph. If $\lambda_v<\frac{1}{d_v-1}$ for all $v\in V$, then
\begin{eqnarray}
T_{\textit{mix}} & \leq & \bar{T}_{\textit{mix}} = \Big\lceil \frac{M}{\theta}\log( n\xi e) \Big\rceil,
\end{eqnarray}
where
\begin{eqnarray}
\theta=\min_{v\in V}\left\{d_v-\sum_{w\in\mathcal N_v}\frac{\lambda_w}{1+\lambda_w}d_w\right\}.
\end{eqnarray}
\end{corollary}

\begin{proof}
Choose $f(v)=\frac{d_v}{q_v}$. By Theorem~\ref{thm:main}, the parallel Glauber dynamics is fast mixing if
$$\theta=\min_{v\in V}\left\{d_v-\sum_{w\in\mathcal N_v}\frac{\lambda_w}{1+\lambda_w}d_w\right\}>0.$$

To achieve that, we need  $\forall v\in V$, $d_v-\sum_{w\in\mathcal N_v}\frac{\lambda_w}{1+\lambda_w}d_w>0.$
It is sufficient to have $\frac{\lambda_w}{1+\lambda_w}d_w<1$, $\forall w\in V$, which is equivalent to
$\lambda_w<\frac{1}{d_w-1}$.
\end{proof}

\begin{remark} \label{remark:boundeddegree}
\emph{Let $\Delta$ be the maximum vertex degree of $G$. If each link (vertex) sends the INTENT message
independently with probability 1/2, then $(1/2)^{\Delta+1} \leq q_v \leq 1$, so $m\geq 1$ and $M\leq \Delta
\cdot 2^{\Delta+1}$. Then we have
$$\bar{T}_{\textit{mix}}  \leq  \Big\lceil \frac{\Delta \cdot 2^{\Delta+1}}{\theta}\log (\Delta \cdot 2^{\Delta+1}en ) \Big\rceil,$$
i.e., the mixing time grows as $O(\log n)$ for bounded-degree graphs when $\lambda_v < \frac{1}{\Delta-1}
\leq \frac{1}{d_v-1}, \forall v$. On the other hand, it was shown in \cite{haysin05} that traditional
single-site Glauber dynamics has a mixing time at least $\Omega(n \log n)$ for bounded-degree graphs.
Therefore, parallel Glauber dynamics reduces the mixing time by an order of magnitude.}
\end{remark}

\begin{remark}
\emph{For interference graphs with special structure, more relaxed conditions can be obtained to ensure fast
mixing. This will be discussed in Section \ref{s.complete-graph}.}
\end{remark}

\section{\label{s.proof-fast}Proof of Fast Mixing}

This section presents the proof of Theorem \ref{thm:main}. In particular, we use the coupling method to
establish the logarithmic mixing time. Readers who are only interested in the main results could skip the
section without loss of continuity.

\subsection{Preliminaries}

A useful technique to bound the mixing time of a Markov chain is via \emph{coupling}.
\begin{definition}
\emph{A \emph{coupling} of the Markov chain is a stochastic process $(X(t), Y(t))$ on $\Omega \times \Omega$
such that $\{X(t)\}$ and $\{Y(t)\}$ marginally are copies the original Markov chain, and if $X(t)=Y(t)$, then
$X(t+1)=Y(t+1)$.}
\end{definition}

Let $\Phi$ be a \emph{distance function} (\emph{metric}) defined on $\Omega \times \Omega$, which satisfies
that for any $x,y,z \in \Omega$: (1) $\Phi(x,y)\geq 0$, with equality if and only if $x=y$; (2) $\Phi(x,y) =
\Phi(y,x)$; (3) $\Phi(x,z) \leq \Phi(x,y)+\Phi(y,z).$ Let
$$D_{min} = \min_{x\neq y}\Phi(x,y), D_{max} =
\max_{x,y}\Phi(x,y), D=\frac{D_{max}}{D_{min}}.$$

The following theorem (e.g., \cite{dyegre00}) can be used to bound the mixing time of the Markov chain.

\begin{theorem} \label{theorem:coupling}
\emph{Suppose $(X(t),Y(t))$ is coupling of the Markov chain where $X(t)$ has distribution $\mu_t$ and $Y(t)$
has distribution $\mu'_t=\pi$. If there exists some constant $\beta<1$ such that, for all $x,y\in \Omega$,}
\begin{equation} \label{eq:betacoupling}
E[\Phi(X(t+1),Y(t+1))|X(t)=x,Y(t)=y]  \leq \beta \Phi(x,y).
\end{equation}
\emph{Then}
\begin{eqnarray} \label{eq:dvbound}
||\mu_t-\pi||_{var} & \leq & \min\{1, \beta^{t}D\}
\end{eqnarray}
\emph{and the mixing time of the Markov chain is bounded by:}
\begin{eqnarray} \label{eq:mixingbound}
T_{mix} \leq \Big\lceil \frac{\log(De)}{\log\beta^{-1}} \Big \rceil \leq \Big\lceil \frac{\log (De)}{1-\beta}
\Big\rceil \doteq \bar{T}_{mix}.
\end{eqnarray}
\end{theorem}

In general, determining $\beta$ in the contraction condition (\ref{eq:betacoupling}) is hard since we need to
check the condition for all pairs of configurations $(x,y)$. In \cite{bubdye97} the so-called \emph{path
coupling} method was introduced by Bubley and Dyer to simplify the calculation. Using path coupling, we only
need to check the contraction condition for certain pairs of configurations. The path coupling method is
described in the following theorem.

\begin{theorem}  \label{theorem:pathcoupling}
\emph{Let $S \subseteq \Omega \times \Omega$ and suppose that for all $x, y \in \Omega \times \Omega$, there
exists a path $x=z_0, z_1, \ldots, z_r=y$ between $x$ and $y$ such that $(z_l,z_{l+1})\in S$ for $0\leq l <r$
and $\Phi(x,y) = \sum_{l=0}^{r-1}\Phi(z_l, z_{l+1}).$ Suppose $(X(t),Y(t))$ is a coupling of the Markov chain
as in Theorem~\ref{theorem:coupling}. If there exists $\beta<1$ such that for any $(x,y) \in S$, the
contraction condition (\ref{eq:betacoupling}) holds, then we have (\ref{eq:dvbound}) and the mixing time of
the Markov chain is bounded as in (\ref{eq:mixingbound}).}
\end{theorem}

Note that the key simplification in the \emph{path coupling theorem} (Theorem~\ref{theorem:pathcoupling}),
compared to the \emph{coupling theorem} (Theorem~\ref{theorem:coupling}), is that the contraction condition
(\ref{eq:betacoupling}) needs to hold only for $(x,y) \in S$, instead of $(x,y) \in \Omega\times\Omega$.

\subsection{Proof of Fast Mixing} \label{subsec:fastmixing}

Now we analyze the mixing time of Glauber dynamics with parallel updates and heterogenous fugacities using
the path coupling technique. We use the following \emph{distance function} between feasible schedules. For
any $\sigma,\eta\in\Omega$, let
\begin{eqnarray} \label{eq:distfun}
\Phi(\sigma,\eta) = \sum_{v}|\sigma_v-\eta_v|f(v) = \sum_{v\in\sigma\bigtriangleup\eta}f(v),
\end{eqnarray}
where $f(v)>0$ is the \emph{weight function} of $v\in  V$ and $\sigma\bigtriangleup\eta
=(\sigma\setminus\eta)\cup(\eta\setminus\sigma)$ is the \emph{symmetric difference} between $\sigma$ and
$\eta$. Note that the distance function is a weighted Hamming distance function and satisfies all the
properties of a metric.

Consider the following \emph{coupling} $(\sigma(t),\eta(t))$: in every time slot both chains select the same
decision schedule and use the same coin toss for every vertex in the decision schedule if that vertex can be
added to both schedules.

Let $E[\Delta\Phi(\sigma,\eta)]$ be the (conditional) expected change of the the distance between the states
of the two Markov chains $\{\sigma(t)\}$ and $\{\eta(t)\}$ after one slot:
\begin{eqnarray*}
E[\Delta\Phi(\sigma,\eta)] & = & E[\Phi(\sigma(t+1), \eta(t+1)|\sigma(t)=\sigma, \eta(t)=\eta] \\
& & - \Phi(\sigma, \eta).
\end{eqnarray*}

For any $\mathbf m\in \mathcal{I}$, let
$$E[\Delta^{\mathbf m}\Phi(\sigma,\eta)]=E[\Delta\Phi(\sigma,\eta)|\mathbf m \mbox{ is the decision schedule}].$$

We say that $\sigma,\eta\in\Omega$ are \emph{adjacent} and we write $\sigma\sim\eta$ if there exists $v\in V$
such that $\sigma$ and $\eta$ differ only at $v$. Let
$$S=\{(\sigma, \eta):\sigma, \eta \in \Omega \mbox{ and
} \sigma\sim\eta\}$$
be the set of all pairs of adjacent schedules. Note that under the distance function
defined in (\ref{eq:distfun}), for all $\sigma, \eta \in \Omega$, we can find a path $\sigma=\tau_0, \tau_1,
\ldots, \tau_{|\sigma \triangle \eta|}=\eta$ between $\sigma$ and $\eta$ such that $(\tau_l,\tau_{l+1})\in S$
for $0\leq l <|\sigma \triangle \eta|$ and
$$\Phi(\sigma, \eta) = \sum_{l=0}^{|\sigma \triangle
\eta|-1}\Phi(\tau_l, \tau_{l+1}).$$

\begin{lemma}\label{lemma:dist_multiset}
\emph{Consider a pair of adjacent schedules $\sigma$ and $\eta$ that differ only at $v$ (so $\Phi(\sigma,
\eta)=f(v)$),}
\begin{eqnarray} \label{eq:dist_dif}
E[\Delta\Phi(\sigma,\eta)] \leq -q_v f(v)+\sum_{w\in\mathcal N_v} q_w\frac{ \lambda_w}{1+\lambda_w}f(w).
\end{eqnarray}
\end{lemma}

\begin{proof}
Without loss of generality, suppose $\sigma_v=0$ and $\eta_v=1.$ Note that $\eta_v=1$ implies that $\eta_w=0$
for all $w\in \mathcal{N}_v$. Since, $\sigma$ and $\eta$  differ only at $v$, this also means that
$\sigma_w=0$ for all $w\in \mathcal{N}_v$. We have
\begin{eqnarray*}
& & E[\Delta\Phi(\sigma,\eta)] \\
& = &  E_{\mathbf m}\Big[E[\Delta^{\mathbf m}\Phi(\sigma,\eta)] \Big] = \sum_{\mathbf m \in
\mathcal{I}}q_{\mathbf m}E[\Delta^{\mathbf
m}\Phi(\sigma,\eta)]\\
& = & \sum_{\mathbf m \in \mathcal{I}}q_{\mathbf m}\sum_{y\in\mathbf m}E[\Delta^y\Phi(\sigma,\eta)] =
\sum_{y\in V}q_y E[\Delta^y\Phi(\sigma,\eta)].
\end{eqnarray*}

Note that only updates on vertices $v$ and $w \in \mathcal{N}_v$ can affect the value of
$E[\Delta\Phi(\sigma,\eta)]$ because updates on other vertices will have the same outcomes due to coupling.

If $v$ is selected for update and since we use the same coin toss for both Markov chains, then
$\sigma(t+1)=\eta(t+1)$ and $\Phi(\sigma(t+1),\eta(t+1))=0$. Thus $E[\Delta^v\Phi(\sigma,\eta)]=-f(v)$.

If $w\in \mathcal{N}_v$ is selected for update, under schedule $\eta,$ $w$ can only take value $0$ because
$w$ has a neighbor (i.e., $v$) belongs to $\eta$. While under schedule $\sigma,$ there are
two cases: \\
(1) if $w$ has a neighbor in $\sigma$, then $w$ can only take value $0$; \\
(2) if $w$ has no neighbors in $\sigma$, $w$ can take value $1$ with probability
$\frac{\lambda_w}{1+\lambda_w}$ and value $0$ otherwise.

Hence for $w\in \mathcal{N}_v$, $E[\Delta^w\Phi(\sigma,\eta)] \leq \frac{\lambda_w}{1+\lambda_w}f(w).$
Summing up all contributions we have (\ref{eq:dist_dif}).
\end{proof}

Now we are ready to prove Theorem \ref{thm:main}.

\begin{proof} (Theorem \ref{thm:main})
For any pair of adjacent schedules $(\sigma, \eta)\in S$ that differ only at some vertex $v\in V$, from
(\ref{eq:dist_dif}) and (\ref{eq:theta}) we have:
$$E[\Delta\Phi(\sigma,\eta)]  \leq  -\theta \leq
-\frac{\theta}{M}\Phi(\sigma,\eta),$$
where we use the fact that $\Phi(\sigma,\eta) = f(v) \leq M.$
Therefore,
$$E[\Phi(\sigma(t+1),\eta(t+1))|\sigma(t)=\sigma,\eta(t)=\eta]\leq\big(1-\frac{\theta}{M}\big)\Phi(\sigma,\eta).$$
By applying Theorem~\ref{theorem:pathcoupling} where $\beta= 1-\frac{\theta}{M}$ and $D=n\xi$, we have
(\ref{eq:dvboundPGD}) and the bound in (\ref{eq:mixingPGB}).
\end{proof}

\section{Low Delay CSMA Scheduling for Wireless Networks}  \label{sec:queuelength}

In this section we analyze the delay performance of PGD-CSMA. In time slot $t$, let $a_i(t)$ be the number of
packets arriving at link $i$, $\forall i$. Assume arrivals occur at the beginning of a time slot and are
i.i.d. with mean arrival rate $E[a_i(t)]=\nu_i \leq 1$. For simplicity, we assume Bernoulli arrivals, i.e.,
$a_{i}(t)\in\{0,1\}$ with $\Pr(a_{i}(t)=1)=\nu_i$. Note that our results hold for more general arrival
processes as long as the second moment $E[a_i^2(t)]<\infty$, $\forall i$.

Let $Q_i(t)$ be the queue length of link $i$ at the end of time slot $t$. Then it has the following dynamics:
\begin{eqnarray}
Q_i(t+1) =  [Q_i(t)+a_i(t+1)-\sigma_i(t+1)]_+ \label{eq:q_dynamics}
\end{eqnarray}
where $[Q]_+=Q$ if $Q\geq 0$ and $[Q]_+=0$ otherwise.

The \emph{capacity region} of the network is the set of all arrival rates $\boldsymbol \nu$ for which there
exists a scheduling algorithm that can stabilize the queues, i.e., the queues are bounded in some appropriate
stochastic sense depending on the arrival model used. In this paper, stability refers to the \emph{positive
recurrence} of the Markov chain. It is known (e.g., \cite{taseph92}) that the capacity region is given by
\begin{equation} \label{eq:capacity}
\Lambda^{o} = \{\boldsymbol{\nu} \geq\mathbf{ 0} \mbox{ }|\mbox{ }\exists \boldsymbol \mu\in
Co(\Omega),\boldsymbol{\nu} < \boldsymbol{\mu}\},
\end{equation}
where $Co(\Omega)$ is the \emph{convex hull} of the set of feasible schedules, i.e., $\boldsymbol \mu \in
Co(\Omega)$ if $\boldsymbol \mu = \sum_{\sigma\in \Omega}t_{\sigma}\sigma$, where $\sum_{\sigma}t_{\sigma}=1$
and $t_{\sigma}\geq 0$ can be viewed as the fraction of time that schedule $\sigma$ is used. When dealing
with vectors, inequalities are interpreted component-wise. We also define
$$\Lambda := \{\boldsymbol{\nu} \geq\mathbf{ 0} \mbox{ }|\mbox{ }\boldsymbol{\nu} \in Co(\Omega)\}.$$

We use the parallel Glauber dynamics to generate the transmission schedule $\sigma(t)$ in every time slot. We
will show that a small mixing time of PGD implies a small average queue length in a wireless network under
PGD-CSMA. By Little's law this also implies a small average delay in the network.

\subsection{Throughput and Fugacities}

The following theorem is slightly extended from \cite{jiawal08}.
\begin{theorem}
\emph{Given any ${\boldsymbol \nu}\in\Lambda^{o}$, there exist suitable fugacities ${\boldsymbol \lambda}$
such that for every link $i$, its mean service rate $s_{i}:=\sum_{\sigma:i\in\sigma}\pi(\sigma)$ is equal to
the mean arrival rate $\nu_{i}$ under PGD-CSMA with fugacities ${\boldsymbol \lambda}$, where $\pi(\sigma)$
is given in (\ref{eq:productformPGD}). In other words, the service rates of PGD-CSMA can exactly meet the
arrival rates at all links. Denote $r_{i}^{*}:=\log(\lambda_{i})$. The vector ${\bf r}^{*}=(r_{i}^{*})$ is
the solution of the convex optimization problem $\max_{{\bf r}}F({\bf r};{\boldsymbol \nu})$ where
\begin{equation}
F({\bf
r};{\boldsymbol{\nu}})=\sum_{i}\nu_{i}r_{i}-\log(\sum_{\sigma\in\Omega}\exp(\sum_{i}\sigma_{i}r_{i})).\label{eq:F}
\end{equation}}
\end{theorem}

\begin{remark}
\emph{A way to understand this result is that $\partial F({\bf r}^{*};{\boldsymbol \nu})/\partial
r_{i}=\nu_{i}-s_{i}({\bf r}^{*})=0,\forall i,$ where $s_{i}({\bf r}^{*})$ is the mean service rate of link
$i$ with the fugacity $\lambda_{j}=\exp(r_{j}^{*}), \forall j$.}
\end{remark}

Now we characterize the suitable fugacities ${\boldsymbol \lambda}$ when ${\boldsymbol \nu}$ is in a fraction
of the capacity region.

\begin{definition}
\emph{The interference degree $\chi_i$ of a link $i$ is the maximum number of links in its conflict set
$\mathcal{N}_i$ that can be active simultaneously. The interference degree of $G$ is defined as
$\chi=\max_{i\in V}\chi_i$.}
\end{definition}


\begin{lemma} \label{lemma:fugacitybound}
\emph{For an arrival rate vector $\boldsymbol{\nu}\in \Lambda^{o}$, let
$\boldsymbol{\lambda}(\boldsymbol{\nu})$ be the vector of fugacities such that the mean service rate $s_i$ is
equal to the mean arrival rate $\nu_i$ for every link $i$ under PGD-CSMA. If $\boldsymbol{\nu} \in
\rho\Lambda^{o}$ (which means that $\frac{1}{\rho} \boldsymbol{\nu} \in \Lambda^{o}$) for some $\rho \leq
\frac{1}{\chi}$, then}
\begin{eqnarray}
\lambda_i(\boldsymbol{\nu}) \leq \frac{\rho}{1-\rho}, \mbox{ }\forall i.
\end{eqnarray}
\end{lemma}

\begin{proof}
For easy notation we write $\boldsymbol{\lambda}(\boldsymbol{\nu})$ as $\boldsymbol{\lambda}$ in this proof.
Under PGD-CSMA with fugacities $\boldsymbol{\lambda}$, let $p_{i,0}$ be the steady-state probability that no
link in $\mathcal{N}_i$ is transmitting, i.e.,
$$p_{i,0}=\sum_{\sigma:\sigma_j=0,\forall j\in \mathcal{N}_i}
\pi(\sigma).$$ Using (\ref{eq:productformPGD}), it is not difficult to show that (a detailed derivation is
given in Appendix~\ref{proof:MRF})
\begin{eqnarray}
\nu_i=s_i=\frac{\lambda_i}{1+\lambda_i} p_{i,0}. \label{eq:MRF}
\end{eqnarray}

Since $1-p_{i,0}$ is the probability that at least one link in $\mathcal{N}_i$ is transmitting, a union bound
yields
\begin{eqnarray} \label{eq:pi0}
1-p_{i,0} \leq \sum_{j\in \mathcal{N}_i}s_j=\sum_{j\in \mathcal{N}_i}\nu_j.
\end{eqnarray}

On the other hand, note that $\boldsymbol{\nu}':=\frac{1}{\rho}\boldsymbol{\nu}\in \Lambda^{o}$, hence there
exists a scheduling algorithm which can serve $\boldsymbol{\nu}'$. Under that scheduling algorithm,
$1-\nu'_i$ is the fraction of time link $i$ is idle. The links in $\mathcal{N}_i$ can be served only when
link $i$ is idle, and in this case at most $\chi$ of them can be served. Therefore,
\begin{eqnarray} \label{eq:nuj}
\sum_{j\in \mathcal{N}_i}\nu'_{j} \leq \chi(1-\nu'_i).
\end{eqnarray}

Combining (\ref{eq:pi0}) and (\ref{eq:nuj}) we have
\begin{eqnarray}
1-p_{i,0} \leq \sum_{j\in \mathcal{N}_i}\nu_j \leq \rho \chi (1-\frac{\nu_i}{\rho}) \leq 1-
\frac{\nu_i}{\rho}
\end{eqnarray}
since $\rho \leq \frac{1}{\chi}$. Hence $\nu_i \leq \rho p_{i,0}$ which implies
$\frac{\lambda_i}{1+\lambda_i} \leq \rho$ and $\lambda_i \leq \frac{\rho}{1-\rho}$.
\end{proof}

Note that for any graph its interference degree $\chi$ is smaller than or equal to its maximum vertex degree
$\Delta$, hence we have:
\begin{corollary} \label{coro:fugacities}
\emph{If the arrival rate vector $\boldsymbol{\nu} \in \frac{1}{\Delta}\Lambda^{o}$, then}
\begin{eqnarray}
\lambda_i(\boldsymbol{\nu}) \leq \frac{1}{\Delta-1}, \mbox{ }\forall i.
\end{eqnarray}
\end{corollary}

\subsection{Delay Performance of Fixed-Parameter PGD-CSMA}

Consider a wireless network with $n$ links and suppose that the maximum degree of the interference graph is
$\Delta$ which is independent of $n$. The queue length of link $i$, $Q_i(t)$, follows the dynamics in
(\ref{eq:q_dynamics}).

\begin{theorem}
\emph{If the arrival rate vector $\boldsymbol{\nu} \in \rho\Lambda^{o}$ for some constant
$\rho<\frac{1}{\Delta}$, then there exist fugacities $\boldsymbol{\lambda}$ such that in the steady state,
the expected queue length $E[Q_i(t)]=O(\log n)$ under PGD-CSMA with $\boldsymbol{\lambda}$.}
\end{theorem}
\begin{proof}
First, we have $\frac{1}{\Delta+1}{\bf 1}\in \Lambda$. This is a direct consequence of Theorem 1 in
\cite{Sufficient}. Choose a constant $\epsilon'\in (0,\frac{1}{\Delta}-\rho)$. Then
$\frac{\epsilon'}{\Delta+1}{\bf 1}\in \epsilon'\Lambda$. Next we state a lemma used in \cite{shatsetsi09}.
\begin{lemma} \label{lemma:linearcap}
\emph{If ${\boldsymbol \nu}_{a}\in\rho_{a}\Lambda$ and ${\boldsymbol \nu}_{b}\in\rho_{b}\Lambda$ , then
${\boldsymbol \nu}_{a}+{\boldsymbol \nu}_{b}\in(\rho_{a}+\rho_{b})\Lambda$.}
\end{lemma}

Since $\boldsymbol{\nu} \in \rho\Lambda^{o}$, there exists $\boldsymbol{\mu}>\boldsymbol{\nu}$ such that
$\boldsymbol{\mu} \in \rho\Lambda$. By Lemma \ref{lemma:linearcap}, we have $\boldsymbol{\mu} +
\frac{\epsilon'}{\Delta+1}{\bf 1}\in (\rho+\epsilon')\Lambda$. So $\boldsymbol{\nu} +
\frac{\epsilon'}{\Delta+1}{\bf 1}\in (\rho+\epsilon')\Lambda^{o} \subset \frac{1}{\Delta}\Lambda$.

Let $\boldsymbol{\lambda}$ be the fugacities such that $s_i=\nu_i+\frac{\epsilon'}{\Delta+1}>\nu_i,\forall
i$, where $s_i$ is the mean service rate for link $i$ under PGD-CSMA with (fixed) fugacities
$\boldsymbol{\lambda}$.\footnote{In order to find $\boldsymbol{\lambda}$ that yields the desired service
rates $s_i$'s, we can use the (fully-distributed) adaptive CSMA algorithms with diminishing step sizes as
proposed in \cite{liuyipro08,CDC_convergence,CSMA_IT}. In these algorithms, the links dynamically adjust the
fugacities and make them converge to the proper $\boldsymbol{\lambda}$. In this subsection, we quantify the
expected queue lengths after such $\boldsymbol{\lambda}$ is found and fixed.} By Lemma
\ref{lemma:fugacitybound}, we know that $\lambda_i \le (\rho+\epsilon')/[1-(\rho+\epsilon')] <
1/(\Delta-1),\forall i$. Therefore the PGD has a mixing time of $O(\log(n))$.

From (\ref{eq:q_dynamics}) and since $a_i(t+1)\geq 0$, we have
\begin{eqnarray}
Q_i(t+1)  \leq  [Q_i(t)-\sigma_i(t+1)]_++a_i(t+1).
\end{eqnarray}
More generally, consider $T$ consecutive time slots beginning from slot $t$, we have
\begin{eqnarray} \label{eq:qT}
Q_i(t+T) & \leq & [Q_i(t)-\sum_{k=1}^{T}\sigma_i(t+k)]_+ + \sum_{k=1}^{T}a_i(t+k) \nonumber \\
& = & [Q_i(t)-T\hat{s}_i]_+ + T \hat{a}_i,
\end{eqnarray}
where $\hat{a}_i=\frac{1}{T}\sum_{k=1}^{T}a_i(t+k)$ and $\hat{s}_i=\frac{1}{T}\sum_{k=1}^{T}\sigma_i(t+k)$
are the average arrival and service rates during the $T$ slots. The RHS of (\ref{eq:qT}) can be viewed as the
\emph{virtual} queue length after $T$ slots if we assume arrivals to link $i$ during the $T$ slots occur at
the end of time slot $t+T$, which is clearly larger than or equal to the actual queue length $Q_i(t+T)$.

From (\ref{eq:qT}) and note that $0\leq \hat{a}_i, \hat{s}_i\leq 1$, we have
\begin{eqnarray} \label{eq:Q^2}
Q_i^2(t+T) & \leq & \Big([Q_i(t)-T\hat{s}_i]_+ + T \hat{a}_i\Big)^2 \nonumber \\
& \leq & [Q_i(t)-T\hat{s}_i]^2 + 2T\hat{a}_iQ_i(t) + T^2 \hat{a}^2_i \nonumber \\
& = & Q_i^2(t) +2TQ_i(t)\Big(\hat{a}_i-\hat{s}_i\Big) + T^2(\hat{a}^2_i+\hat{s}^2_i). \nonumber\\
& \leq & Q_i^2(t) +2TQ_i(t)\Big(\hat{a}_i-\hat{s}_i\Big) + 2T^2.
\end{eqnarray}

Note that $(\mathbf{Q}(t), \sigma(t))$ evolves as a Markov chain. Define the Lyapunov function
$L(t)=\frac{1}{2}\sum_{i}Q_i^2(t)$. From (\ref{eq:Q^2}) and since $E[Q_i(t)\hat{s}_i|\mathbf{Q}(t),\sigma(t)]
= Q_i(t)E[\hat{s}_i|\sigma(t)]$, we have
\begin{eqnarray} \label{eq:drift}
& & E[L(t+T)-L(t)|\textbf{Q}(t),\sigma(t)] \nonumber \\
& \leq & T\sum_{i}Q_i(t)(\nu_i-E[\hat{s}_i|\sigma(t)])+nT^2.
\end{eqnarray}

Now let us compute $E[\hat{s}_{i}|\sigma(t)]$ for every link $i$:
\begin{eqnarray*}
 E[\hat{s}_{i}|\sigma(t)] & = & \frac{1}{T}\sum_{k=1}^{T}E[\sigma_{i}(t+k)|\sigma(t)]\\
& = & \frac{1}{T}\sum_{k=1}^{T}\Pr(\sigma_{i}(t+k)=1|\sigma(t)) \\
& = & \frac{1}{T}\sum_{k=1}^{T}\sum_{\sigma\in\Omega:\sigma_{i}=1}\mu_{\mathbf \sigma(t),k}(\sigma)
\end{eqnarray*}
where $\mu_{\sigma(t),k}$ is the distribution of the Markov chain of the schedules after $k$ slots if the
Markov chain starts with schedule $\sigma(t)$. Remember $s_{i} =
\sum_{\sigma\in\Omega:\sigma_{i}=1}\pi(\sigma)$, so
\begin{eqnarray*}
& &  |E[\hat{s}_{i}|\sigma(t)]-s_i| \\
& = & | \frac{1}{T}\sum_{k=1}^{T} \Big( \sum_{\sigma\in\Omega:\sigma_{i}=1}\mu_{\mathbf \sigma(t),k}(\sigma)-
\sum_{\sigma\in\Omega:\sigma_{i}=1}\pi(\sigma) \Big)| \\
& \leq & \frac{1}{T}\sum_{k=1}^{T}|\sum_{\sigma\in\Omega:\sigma_{i}=1}\mu_{\sigma(t),k}(\sigma)-\sum_{\sigma\in\Omega:\sigma_{i}=1}\pi(\sigma)|\\
& \leq & \frac{1}{T}\sum_{k=1}^{T}||\mu_{\sigma(t), k}-\pi||_{var}.
\end{eqnarray*}

Since $\lambda_i < \frac{1}{\Delta-1}$, from Theorem~\ref{thm:main} and Corollary~\ref{coro:c3} we know that
the parallel Glauber dynamics under $\boldsymbol{\lambda}$ is fast mixing, and we can find $D=n\xi$ and
$\beta=1-\frac{\theta}{M}<1$ such that
\begin{eqnarray*}
||\mu_{\sigma(t), k}-\pi||_{var} \leq \min\{1,D\beta^k\}.
\end{eqnarray*}

Let $T_0 = \lfloor  \frac{\log D}{\log \beta^{-1}}\rfloor$ so $D\beta^{T_0+1} \leq 1$, we have
\begin{eqnarray}
 |E[\hat{s}_{i}|\sigma(t)]-s_i|  & \leq & \frac{1}{T}\sum_{k=1}^{T_0} 1 + \frac{1}{T}\sum_{k=T_0+1}^{T} D \beta^{k} \nonumber \\
 & \leq & \frac{T_0+\frac{1}{1-\beta}}{T} \leq \frac{\bar{T}_{mix}}{T} \label{eq:mix_diff}
\end{eqnarray}
where $\bar{T}_{mix} \doteq \Big\lceil \frac{\log (De)}{1-\beta} \Big\rceil \geq T_0+\frac{1}{1-\beta}.$ If
we choose $T=\lceil \frac{\bar{T}_{mix}}{\epsilon} \rceil$ for some $\epsilon>0$, then
$|E[\hat{s}_{i}|\sigma(t)]-s_i| \leq \epsilon$ and
\begin{eqnarray} \label{eq:sibound}
E[\hat{s}_i|\sigma(t)] \geq s_i-\epsilon, \forall i.
\end{eqnarray}
Plugging into (\ref{eq:drift}), we have
\begin{eqnarray*}
& & E[L(t+T)-L(t)|\textbf{Q}(t),\sigma(t)]  \\
& \leq & -T \sum_{i}Q_i(t)(s_i-\nu_i-\epsilon) + nT^2
\end{eqnarray*}
where the right-hand-side is negative if $\epsilon < \min_{i}(s_i-\nu_i)$ and $||\mathbf{Q}(t)||$ is
sufficiently large. This establishes the \emph{negative drift} of $L(t)$. By the Foster-Lyapunov criterion,
$(\mathbf{Q}(t),\sigma(t))$ is positive recurrent.

In the steady state, taking expectations of both sides of (\ref{eq:Q^2}) and using (\ref{eq:sibound}), we
have
\begin{eqnarray*}
E[Q_i^2(t+T) - Q_i^2(t)] = 0 \leq  -2T E[Q_i(t)](s_i-\nu_i-\epsilon) + 2T^2
\end{eqnarray*}
which implies
\begin{eqnarray}
E[Q_i(t)]  \leq  \frac{T}{s_i-\nu_i-\epsilon} \leq \frac{\bar{T}_{mix}+1}{\epsilon(s_i-\nu_i-\epsilon)} =
\frac{4(\bar{T}_{mix}+1)}{(s_i-\nu_i)^2}\label{eq:per-queue-bound}
\end{eqnarray}
by choosing $\epsilon=\frac{s_i-\nu_i}{2}$ and $T=\lceil \frac{\bar{T}_{mix}}{\epsilon} \rceil$ in our
analysis of link $i$. Recall that $s_i-\nu_i=\epsilon'/(\Delta+1)$ is independent of $n$. Also, as we have
proved in Section~\ref{subsec:fastmixing}, for bounded-degree interference graphs, $\bar{T}_{mix}=O(\log n)$.
These facts combined imply that $E[Q_i(t)]=O(\log n)$ under PGD-CSMA.
\end{proof}

\subsection{Delay Performance of Dynamic-Parameter PGD-CSMA}

In this subsection we consider PGD-CSMA with dynamic parameters (fugacities). That is, unlike the last
subsection where we assumed that suitable fugacities have been found and fixed, here the fugacities are
dynamically adjusted based on the local queue length information.

Given an interference graph $G$ with $n$ links. Suppose $G$ has a maximum degree $\Delta$ and an interference
degree $\chi$ which are all independent of $n$.

Let $B$ be such that $\exp(B)\le\frac{1}{\chi-1}$. Assume that ${\boldsymbol
\nu}\in\frac{\exp(B_{\epsilon})}{1+\exp(B_{\epsilon})}\Lambda^{o}$ where $B_{\epsilon}:=B-\epsilon \in
(0,B)$, and that each element $\nu_{k}\ge\nu_{min}$ for some constant $\nu_{min}>0$.

\begin{proposition}
\label{pro:lambda}\emph{Under the above assumptions of ${\boldsymbol \nu}$, we have}
\[\nu_{min} \le \lambda_i({\boldsymbol \nu}) \le \exp (B_{\epsilon}),\forall i. \]
\end{proposition}
\begin{proof}
By Lemma \ref{lemma:fugacitybound}, we have $\lambda_i({\boldsymbol \nu}) \le \exp (B_{\epsilon})$. Also, by
(\ref{eq:MRF}), one has $\nu_i=s_i\le \lambda_i({\boldsymbol \nu})/[1+\lambda_i({\boldsymbol \nu})]$. Since
$\nu_i\ge \nu_{min}$, we have $\lambda_i({\boldsymbol \nu})\ge \nu_{min}$.
\end{proof}
\begin{remark}
\emph{Let ${\bf r}^*=\arg\max_{{\bf r}}F({\bf r};{\boldsymbol \nu})$ where $F({\bf r};{\boldsymbol \nu})$ is
defined in (\ref{eq:F}). Since $r^{*}_i=\log (\lambda_i ({\boldsymbol \nu}))$, we have $r^{*}_i \in
[r_{min},B_{\epsilon}],\forall i$, where $r_{min}:=\log (\nu_{min})$.}
\end{remark}

We further select $B$ such that if $\lambda_{k}\le\exp(B),\forall k$, then the parallel Glauber dynamics is
fast mixing and the mixing time is upper-bounded by $\bar{T}_{mix}=O(\log n)$ (by Corollary~\ref{coro:c3},
this can be achieved when $\exp(B)<\frac{1}{\Delta-1}$). Now we propose an algorithm to dynamically adjust
the fugacities of PGD-CSMA.

\medskip{}

\textbf{Algorithm 1:} The fugacities of the links, denoted by the vector ${\boldsymbol \lambda}$, are updated
every $T$ time slots, where
\begin{equation} T=\left\lceil
\bar{T}_{mix}\cdot\frac{4n\cdot(B-r_{min}+\alpha)}{\delta}\right\rceil =O(n\log n) \label{eq:T}
\end{equation}
with $\alpha:=\delta/n$ and
\begin{equation}
\delta:=\int_{r=B_{\epsilon}}^{B}[\frac{\exp(r)}{1+\exp(r)}-\frac{\exp(B_{\epsilon})}{1+\exp(B_{\epsilon})}]dr>0.
\label{eq:delta}
\end{equation}

Specifically, at the end of slot $jT,j=0,1,\dots$, link $k$ updates its fugacity to be
\begin{equation}
\lambda_{k}[j]=\exp(\min\{r_{k}[j],B\}),\forall k \label{eq:lambda_update}
\end{equation} with
\begin{equation}
r_{k}[j]:=\frac{\alpha}{T}Q_{k}[j]+r_{min}-\alpha\label{eq:r_update}
\end{equation}
where $Q_{k}[j]$ is the queue length of link $k$ at the end of slot $jT$, i.e., $Q_{k}[j]=Q_{k}(jT)$. Note
that $Q_{k}(t)$ follows the dynamics (\ref{eq:q_dynamics}). Also note that the fugacity vector ${\bf
\lambda}[j]$ is used for $T$ time slots (from slot $jT+1$ to $(j+1)T$, which we call {}``frame $j$'').

\begin{remark}
\emph{By (\ref{eq:lambda_update}) and (\ref{eq:r_update}), link $k$ increases its fugacity when its queue
length increases (unless the fugacity has reached $\exp (B)$). So link $k$ transmits more aggressively when
its queue builds up. Also, since $\lambda_{k}[j]\le \exp (B), \forall k$ at all time, PGD is fast mixing in
each frame of $T$ slots.}
\end{remark}
\begin{remark}
\emph{Algorithm 1 is designed such that ${\bf r}[j]$ is attracted towards ${\bf r}^* \in
[r_{min},B_{\epsilon}]^n$. So, by (\ref{eq:r_update}), the queue lengths are attracted towards an affine
function of ${\bf r}^{*}$ and are therefore stabilized; and by (\ref{eq:lambda_update}), ${\boldsymbol
\lambda}[j]$ is attracted to ${\boldsymbol \lambda}({\boldsymbol \nu})$.}
\end{remark}

\medskip{}

Algorithm 1 and the consequent proof techniques are quite different from existing works (e.g.,
\cite{CSMA_IT}). Specifically, unlike \cite{CSMA_IT}, the fugacities in Algorithm 1 are direct functions of
the queue lengths. Also, we derive a polynomial delay bound instead of an exponential bound in
\cite{CSMA_IT}. To do that, we apply the mixing time results in Section~\ref{sec:mixingPGD} and use a novel
Lyapunov function in the stability proof.

\begin{theorem}\label{thm:q_stable}
\emph{The queue length vector ${\bf Q}(t)$ is stable (i.e., positive recurrent) under Algorithm 1.}
\end{theorem}

We first need a lemma.
\begin{lemma}\label{lem:gap}
\emph{For any vector ${\bf r}$ with some element $r_{k}\ge B$, we have $F({\bf r};{\boldsymbol \nu})\le
F({\bf r}^{*};{\boldsymbol \nu})-\delta$, where $\delta>0$ is defined in (\ref{eq:delta}). Note that $\delta$
is independent of $n$.}
\end{lemma}
\begin{proof}
We first show that $F^(\tilde{{\bf r}};{\boldsymbol \nu})\le F({\bf r}^{*};{\boldsymbol \nu})-\delta$ if
$\tilde{r}_{k}=B$ for some $k$. Denote
\begin{equation}
F_{k}(\bar{r}_{k};{\boldsymbol \nu})=\max_{{\bf r}_{-k}}F({\bf r}_{-k},r_{k}=\bar{r}_{k};{\boldsymbol
\nu}).\label{eq:max-func}
\end{equation}
Then clearly $F(\tilde{{\bf r}};{\boldsymbol \nu})\le F_{k}(B;{\boldsymbol \nu})$. So it is sufficient to
prove \begin{equation} F_{k}(B;{\boldsymbol \nu})\le F({\bf r}^{*};{\boldsymbol
\nu})-\delta.\label{eq:suff}\end{equation}

Denote the solution of RHS of (\ref{eq:max-func}) by $\hat{{\bf r}}_{-k}(\bar{r}_{k})$, and let $\hat{{\bf
r}}(\bar{r}_{k}):=(r_{k}=\bar{r}_{k},{\bf r}_{-k}=\hat{{\bf r}}_{-k}(\bar{r}_{k}))$. Then, the envelope
theorem implies that
\begin{equation}
dF_{k}(\bar{r}_{k};{\boldsymbol \nu})/dr_{k}=\partial F(\hat{{\bf r}}(\bar{r}_{k});{\boldsymbol
\nu})/\partial r_{k}=\nu_{k}-s_{k}(\hat{{\bf r}}(\bar{r}_{k})). \label{eq:derivative-max-func}
\end{equation}

By the definition of $\hat{{\bf r}}_{-k}(\bar{r}_{k})$, we know that for any $k'\ne k$, $\partial F(\hat{{\bf
r}}(\bar{r}_{k});{\boldsymbol \nu})/\partial r_{k'}=\nu_{k'}-s_{k'}(\hat{{\bf r}}(\bar{r}_{k}))=0$, so
$s_{k'}(\hat{{\bf r}}(\bar{r}_{k}))=\nu_{k'}$. Therefore, ${\bf s}(\hat{{\bf r}}(\bar{r}_{k}))-{\boldsymbol
\nu}=(s_{k}(\hat{{\bf r}}(\bar{r}_{k}))-\nu_{k})\cdot{\bf e}_{k}$.

Note that ${\boldsymbol \nu}\in\frac{\exp(B_{\epsilon})}{1+\exp(B_{\epsilon})}\Lambda$. Given a
$\bar{r}_{k}\in(B_{\epsilon},B]$, by Lemma \ref{lemma:fugacitybound} we know that ${\bf s}(\hat{{\bf
r}}(\bar{r}_{k}))\notin\rho\Lambda$ for any $\rho<\frac{\exp(\bar{r}_{k})}{1+\exp(\bar{r}_{k})}$. So
$s_{k}(\hat{{\bf r}}(\bar{r}_{k}))-\nu_{k}>0$.

Since $(s_{k}(\hat{{\bf r}}(\bar{r}_{k}))-\nu_{k})\cdot{\bf e}_{k}\in(s_{k}(\hat{{\bf
r}}(\bar{r}_{k}))-\nu_{k})\Lambda$, using Lemma~\ref{lemma:linearcap} we have ${\bf s}(\hat{{\bf
r}}(\bar{r}_{k}))={\boldsymbol \nu}+(s_{k}(\hat{{\bf r}}(\bar{r}_{k}))-\nu_{k})\cdot{\bf
e}_{k}\in[\frac{\exp(B_{\epsilon})}{1+\exp(B_{\epsilon})}+s_{k}(\hat{{\bf r}}(\bar{r}_{k}))-\nu_{k}]\Lambda$.
Since ${\bf s}(\hat{{\bf r}}(\bar{r}_{k}))\notin\rho\Lambda$ for any
$\rho<\frac{\exp(\bar{r}_{k})}{1+\exp(\bar{r}_{k})}$, it must be that $s_{k}(\hat{{\bf
r}}(\bar{r}_{k}))-\nu_{k}\ge\frac{\exp(\bar{r}_{k})}{1+\exp(\bar{r}_{k})}-\frac{\exp(B_{\epsilon})}{1+\exp(B_{\epsilon})}$.
Using (\ref{eq:derivative-max-func}), one has $dF_{k}(\bar{r}_{k};{\boldsymbol
\nu})/dr_{k}\le\frac{\exp(B_{\epsilon})}{1+\exp(B_{\epsilon})}-\frac{\exp(\bar{r}_{k})}{1+\exp(\bar{r}_{k})}$.

Therefore, we have\begin{eqnarray*}
& & F_{k}(B;{\boldsymbol \nu}) \\
& = & F_{k}(B_{\epsilon};{\boldsymbol \nu})+\int_{\bar{r}_{k}=B_{\epsilon}}^{B}[\frac{\exp(B_{\epsilon})}{1+\exp(B_{\epsilon})}-\frac{\exp(\bar{r}_{k})}{1+\exp(\bar{r}_{k})}]d\bar{r}_{k}\\
 & \le & F({\bf r}^{*};{\boldsymbol \nu})-\int_{\bar{r}_{k}=B_{\epsilon}}^{B}[\frac{\exp(\bar{r}_{k})}{1+\exp(\bar{r}_{k})}-\frac{\exp(B_{\epsilon})}{1+\exp(B_{\epsilon})}]d\bar{r}_{k}.\end{eqnarray*}

So, (\ref{eq:suff}) holds with $\delta$ defined in (\ref{eq:delta}).

Finally, we need to show that $F(\tilde{{\bf r}};{\boldsymbol \nu})\le F({\bf r}^{*};{\boldsymbol
\nu})-\delta$ if $\tilde{r}_{k}>B$ for some $k$. Given such a $\tilde{{\bf r}}$, one can find a $\tilde{{\bf
r}}'$ on the line segment between $\tilde{{\bf r}}$ and ${\bf r}^{*}$ to satisfy $\tilde{r}_{k}'=B$. We
already know that $F(\tilde{{\bf r}}';{\boldsymbol \nu})\le F({\bf r}^{*};{\boldsymbol \nu})-\delta$. By the
concavity of $F(\cdot;{\boldsymbol \nu})$, we have $F(\tilde{{\bf r}};{\boldsymbol \nu})\le F({\bf
r}^{*};{\boldsymbol \nu})-\delta$.
\end{proof}

Now we are ready to prove Theorem \ref{thm:q_stable}.

\begin{proof} (Theorem~\ref{thm:q_stable})
We first show that ${\bf r}[j]$ is stable. Equation (\ref{eq:r_update}) implies that $r_{k}[j]\ge
r_{min}-\alpha,\forall k,j$. Denote
$$r'_{k}[j]:=\min\{r_{k}[j],B\}.$$
Then $\lambda_{k}[j]=\exp(r'_{k}[j])$, and ${\bf r}'[j]\in\tilde{{\cal B}}:=[r_{min}-\alpha,B]^{n},\forall
j.$

We define a \emph{Lyapunov function}
\begin{equation}
L({\bf r}):=\sum_{k}L_{k}(r_{k})\label{eq:Lya}
\end{equation}
where
\begin{eqnarray*}
L_{k}(r_{k}): & = & (B-r_{k}^{*})(r_{k}-r_{k}^{*})I(r_{k}\ge B)+\\
 &  & \frac{1}{2}[(r_{k}-r_{k}^{*})^{2}+(B-r_{k}^{*})^{2}]I(r_{k}<B).
 \end{eqnarray*}
Then,
\begin{equation}
\frac{\partial L({\bf r})}{\partial r_{k}}=(B-r_{k}^{*})I(r_{k}\ge
B)+(r_{k}-r_{k}^{*})I(r_{k}<B).\label{eq:L_derivative}
\end{equation}

\begin{figure}
\includegraphics[width=8cm]{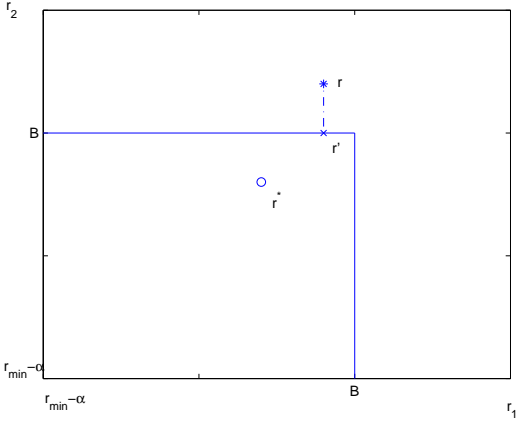}
\caption{An example when $n=2$.} \label{fig:rprojection}
\end{figure}

For simplicity, we write $F({\bf r};{\boldsymbol \nu})$ as $F({\bf r})$. For a ${\bf
r}\in[r_{min}-\alpha,\infty)^{n}$ but ${\bf r}\notin\tilde{{\cal B}}$, let ${\bf r}'$ be its projection on
$\tilde{{\cal B}}$ (i.e., $r'_{k}=\min\{r_{k},B\},\forall k$. See Fig. \ref{fig:rprojection} for an
illustration). Then by the concavity of $F({\bf r})$ we have
\begin{eqnarray}
\sum_{k}\frac{\partial F({\bf r}')}{\partial r_{k}'}\frac{\partial L({\bf r})}{\partial r_{k}} & = &
\sum_{k}\frac{\partial F({\bf r}')}{\partial r_{k}'}(r_{k}'-r_{k}^{*}) \nonumber \\
& \le & F({\bf r'})-F({\bf r}^{*}) \le-\delta\label{eq:strictly-negative}
\end{eqnarray}
where the last step has used Lemma \ref{lem:gap}.

We then claim that given any ${\bf r}[j]\in[r_{min}-\alpha,\infty)^{n}$,
\begin{equation}
L_{k}(r_{k}[j+1])\le L_{k}(\tilde{r}_{k}[j+1]),\forall k\label{eq:Lk_ineq}
\end{equation}
where $\tilde{r}_{k}[j+1]:=r_{k}[j]+\alpha[\hat{a}_{k}[j]-\hat{s}_{k}[j]]$, with $\hat{a}_{k}[j]$ denoting
the average arrival rate from slot $jT+1$ to $(j+1)T$, and $\hat{s}_{k}[j]$ denoting the average service rate
in the same interval with the fugacity vector ${\bf \lambda}[j]$.

To show (\ref{eq:Lk_ineq}) we consider two cases:

(i) If from slot $jT+1$ to $(j+1)T$, queue $k$ is not empty whenever it is scheduled to transmit in the CSMA
protocol, then $Q_{k}[j+1]=Q_{k}[j]+T[\hat{a}_{k}[j]-\hat{s}_{k}[j]]$. By (\ref{eq:r_update}), we have
$r_{k}[j+1]=r_{k}[j]+\alpha[\hat{a}_{k}[j]-\hat{s}_{k}[j]]=\tilde{r}_{k}[j+1]$. Then
$L_{k}(\tilde{r}_{k}[j+1])=L_{k}(r_{k}[j+1])$.

(ii) Otherwise, one has
\begin{equation}
Q_{k}[j+1]>Q_{k}[j]+T[\hat{a}_{k}[j]-\hat{s}_{k}[j]].\label{eq:served-empty-queue}
\end{equation}
Inequality (\ref{eq:qT}) is equivalent to
$Q_{k}[j+1]\le[Q_{k}[j]-T\cdot\hat{s}_{k}[j]]_{+}+T\cdot\hat{a}_{k}[j]$. Suppose that $Q_{k}[j+1]>T$, then
$[Q_{k}[j]-T\cdot\hat{s}_{k}(j)]_{+}\ge Q_{k}[j+1]-T\cdot\hat{a}_{k}[j]>T-T\cdot\hat{a}_{k}[j]\ge0$, which
implies that $Q_{k}[j]>T\cdot\hat{s}_{k}[j]$. But if this holds, queue $k$ never gets empty from slot $jT+1$
to $(j+1)T$, contradicting the assumption. Therefore, we have $Q_{k}[j+1]\le T$. Using this,
(\ref{eq:served-empty-queue}) and (\ref{eq:r_update}), one has\[ \tilde{r}_{k}[j+1]<r_{k}[j+1]\le r_{min},\]
which implies (\ref{eq:Lk_ineq}) since $r_{k}^{*}\ge r_{min}$. This completes the proof of
(\ref{eq:Lk_ineq}).

Inequality (\ref{eq:Lk_ineq}) immediately implies that \begin{equation} L({\bf r}[j+1])\le L(\tilde{{\bf
r}}[j+1]).\label{eq:L_neq}\end{equation}

Let ${\cal F}_{j}$ be the $\sigma$-field generated by $\{{\bf Q}[j'],{\bf r}[j'],{\bf
\sigma}[j']\},j'=0,1,2,\dots,j$ where ${\bf \sigma}[j']={\bf \sigma}(j'T)$ is the state of the CSMA Markov
chain (i.e., the schedule used in time slot $j'T$). In the following, we write the conditional expectation
$E(\cdot|{\cal F}_{j})$ simply as $E_{j}(\cdot)$. Using Taylor expansion,
\begin{eqnarray}
\Delta[j] & := & E_{j}[L({\bf r}[j+1])-L({\bf r}[j])]\nonumber \\
 & \le & E_{j}[L(\tilde{{\bf r}}[j+1])-L({\bf r}[j])]\nonumber \\
 & \le & \alpha\sum_{k}\Big\{[\nu_{k}-E_{j}(\hat{s}_{k}[j])]\frac{\partial L({\bf r}[j])}{\partial r_{k}[j]}\Big\}+\frac{1}{2}n\alpha^{2}\nonumber \\
 & \le & \alpha\sum_{k}\Big\{[\nu_{k}-s_{k}({\bf r}'[j])]\frac{\partial L({\bf r}[j])}{\partial r_{k}[j]}\Big\}+\alpha\sum_{k}\Big\{[s_{k}({\bf r}'[j])\nonumber \\
 &  & -E_{j}(\hat{s}_{k}[j])]\frac{\partial L({\bf r}[j])}{\partial r_{k}[j]}\Big\}+\frac{1}{2}n\alpha^{2}.\label{eq:drift1}\end{eqnarray}

By (\ref{eq:T}) and (\ref{eq:mix_diff}), we have
\begin{equation}
|E_{j}[\hat{s}_{k}[j]]-s_{k}({\bf r}'[j])|\le\delta/(4n\cdot d_{max}),\forall
j,k.\label{eq:mixing-bound}\end{equation} where $d_{max}:=B-r_{min}+\alpha$.

If ${\bf r}[j]\notin\tilde{{\cal B}}$, we use (\ref{eq:drift1}), (\ref{eq:mixing-bound}) and
(\ref{eq:strictly-negative}) to derive the following:
\begin{eqnarray*}
 &  & E_{j}[L({\bf r}[j+1])-L({\bf r}[j])]\\
 & \le & \alpha\sum_{k}\frac{\partial F({\bf r}'[j])}{\partial r_{k}'[j]}\frac{\partial L({\bf r}[j])}{\partial r_{k}[j]}+\alpha\sum_{k}(\frac{\delta}{4n\cdot d_{max}}d_{max})\\
 &  & +\frac{1}{2}n\alpha^{2}\\
 & = & \alpha\sum_{k}\frac{\partial F({\bf r}'[j])}{\partial r_{k}'[j]}\frac{\partial L({\bf r}[j])}{\partial r_{k}[j]}+\alpha\frac{\delta}{4}+\frac{1}{2}n\alpha^{2}\\
 & \le & \alpha(-\delta+\frac{\delta}{4}+\frac{1}{2}n\alpha) = -\alpha\delta/4
 \end{eqnarray*}
which establishes the negative drift of $L({\bf r}[j])$.

By the Foster-Lyapunov criterion, ${\bf r}[j]$ is stable, and by (\ref{eq:r_update}), ${\bf Q}[j]$ is also
stable. Since in each time slot, the change of each queue length is at most 1, we conclude that ${\bf Q}(t)$
is stable.
\end{proof}

\medskip{}

\begin{theorem}\label{thm:mean_q_dynamic}
\emph{Suppose that the arrival rate vector $\boldsymbol{\nu} \in \rho\Lambda^{o}$ where $\rho <
\frac{1}{\Delta}$, and each element $\nu_k\ge \nu_{min}$. Let $B_{\epsilon}=\log (\frac{\rho}{1-\rho})$ so
that $\boldsymbol{\nu} \in \frac{\exp (B_{\epsilon})}{1+\exp (B_{\epsilon})}\Lambda^{o}$. Let $B=\log
(\frac{\rho^{'}}{1-\rho^{'}})>B_{\epsilon}$ where the constant $\rho^{'}\in (\rho,\frac{1}{\Delta})$, so that
if $\lambda_{k}\le\exp(B)=\frac{\rho^{'}}{1-\rho^{'}}<\frac{1}{\Delta-1}, \forall k$, the parallel Glauber
dynamics has a mixing time of $O(\log (n))$.} \emph{Then, the queue lengths under Algorithm 1 satisfy}
\begin{equation}
\sum_{k}\bar{Q}_{k}=O(n^{3}\log n)\label{eq:average-q}
\end{equation}
\emph{where $\bar{Q}_{k}:=\limsup_{M\rightarrow\infty}\sum_{t=1}^{M}E(Q_{k}(t))/M$ is the  average expected
queue length at link $k$.}
\end{theorem}
\begin{proof}
Denote $\bar{L}:=\max_{{\bf r}\in\tilde{{\cal B}}}L({\bf r})$. Then if $L({\bf r})>\bar{L}$, we have ${\bf
r}\notin\tilde{{\cal B}}$. Define\[ G({\bf r}):=[L({\bf r})-\bar{L}]_{+}.\]

Note that $|r_{k}[j+1]-r_{k}[j]|\le\alpha,\forall k$. Also, (\ref{eq:L_derivative}) implies that
$|\frac{\partial L({\bf r})}{\partial r_{k}}|\le d_{max}$. Then
\begin{eqnarray*}
|L({\bf r}[j+1])-L({\bf r}[j])| & \le & n\alpha d_{max} \\
 & = & \delta d_{max}:=c.\end{eqnarray*}

Case 1: If $L({\bf r}[j])-\bar{L}>c$, then $L({\bf r}[j+1])-\bar{L}>0$, and
\begin{eqnarray*}
 &  & G^{2}({\bf r}[j+1])-G^{2}({\bf r}[j])\\
 & = & [L({\bf r}[j+1])-\bar{L}]^{2}-[L({\bf r}[j])-\bar{L}]^{2}\\
 & = & 2[L({\bf r}[j])-\bar{L}][L({\bf r}[j+1])-L({\bf r}[j])]+\\
 &  & [L({\bf r}[j+1])-L({\bf r}[j])]^{2}\\
 & \le & 2G({\bf r}[j])[L({\bf r}[j+1])-L({\bf r}[j])]+c^{2}.
\end{eqnarray*}
Therefore
\begin{eqnarray*}
E_{j}[G^{2}({\bf r}[j+1])-G^{2}({\bf r}[j])] & \le & 2G({\bf r}[j])\Delta[j]+c^{2}\\
 & \le & -G({\bf r}[j])\alpha\delta/2+c^{2}.
 \end{eqnarray*}

Case 2: If $G({\bf r}[j])\le c$, then $0\le G({\bf r}[j+1])\le G({\bf r}[j])+c$. Therefore\begin{eqnarray*}
 &  & E_{j}[G^{2}({\bf r}[j+1])-G^{2}({\bf r}[j])]\\
 & \le & 2cG({\bf r}[j])+c^{2}\le3c^{2}\\
 & \le & -G({\bf r}[j])\alpha\delta/2+c\alpha\delta/2+3c^{2}.\end{eqnarray*}

Combining case 1 and 2, we have\[ E_{j}[G^{2}({\bf r}[j+1])-G^{2}({\bf r}[j])]\le-G({\bf
r}[j])\alpha\delta/2+c\alpha\delta/2+3c^{2}.\]

Taking expectations on both sides yields\[ E[G^{2}({\bf r}[j+1])-G^{2}({\bf r}[j])]\le-E[G({\bf
r}[j])]\alpha\delta/2+c\alpha\delta/2+3c^{2}.\]

Summing the above inequality from $j=0$ to $j=J-1$, and dividing both sides by $J$, we have
\begin{eqnarray*}
& & E[G^{2}({\bf r}[J])-G^{2}({\bf r}[0])]/J \\
& \le & -(\alpha\delta/2)\sum_{j=0}^{J-1}E[G({\bf r}[j])]/J+c\alpha\delta/2+3c^{2}.
\end{eqnarray*}

Therefore, \[ \limsup_{J\rightarrow\infty}\sum_{j=0}^{J-1}E(G({\bf
r}[j]))/J\le\frac{c\alpha\delta/2+3c^{2}}{\alpha\delta/2}=c+6n\frac{c^{2}}{\delta^{2}}=O(n).\]

Note that \[ W({\bf r}):=\sum_{k}(B-r_{k}^{*})(r_{k}-r_{k}^{*})\le L({\bf r})\le G({\bf r})+\bar{L}.\]

So\[ \limsup_{J\rightarrow\infty}\sum_{j=0}^{J-1}E(W({\bf r}[j]))/J\le O(n)+\bar{L}=O(n)\] since
$\bar{L}=\sum_{k=1}^{n} (B-r_k^{*})^2 = O(n)$.

In view of (\ref{eq:r_update}), we then have
\[ \limsup_{J\rightarrow\infty}\sum_{j=0}^{J-1}\sum_{k}(B-r_{k}^{*})E(Q_{k}[j])/J=O(\frac{T}{\alpha}n)=O(n^{3}\log n).\] Since $B-r_{k}^{*} \ge B-B_{\epsilon}=\epsilon,\forall k$, we have
\begin{eqnarray*}
 &  & \limsup_{J\rightarrow\infty}\sum_{j=0}^{J-1}\sum_{k}E(Q_{k}[j])/J\\
 & \le & \frac{1}{\epsilon}\limsup_{J\rightarrow\infty}\sum_{j=0}^{J-1}\sum_{k}(B-r_{k}^{*})Q_{k}[j]/J=O(n^{3}\log n).\end{eqnarray*}

Since in a slot each queue is increased at most by 1,
$\bar{Q}_{k}\le\limsup_{J\rightarrow\infty}\sum_{j=0}^{J-1}E(Q_{k}[j])/J+T$ where $T=O(n \log n)$. Therefore
(\ref{eq:average-q}) holds.
\end{proof}

\section{\label{s.complete-graph}Complete Interference Graphs}

We have shown that in general interference graphs, PGD-CSMA can achieve polynomial queue lengths when the
arrival rate vector lies in $\frac{1}{\Delta}$ of the capacity region. In this section, we will show that the
polynomial-delay region can be further improved for certain interference graphs. We only consider
fixed-parameter PGD-CSMA for simplicity.

Consider a wireless local area network where every mobile station maintains a $1$-hop link to the Access
Point. This is an important scenario in practice, where all links conflict with each other and the
interference graph is a complete graph.

In such a complete interference graph with $n$ links (note that this graph does not have a bounded degree as
$n\rightarrow \infty$), suppose the fugacity of link $i$ is $\lambda_i$. Under PGD-CSMA with fugacities
$\boldsymbol{\lambda}$, the steady-state probability that link $i$ is active is simply
$$\pi(\textbf{e}_i) = \frac{\lambda_i}{Z}=s_i,$$
where $\textbf{e}_i$ is the schedule with only link $i$ active, $Z=1+\sum_{j=1}^{n}\lambda_j$, and $s_i$ is
the mean service rate of link $i$. Also, denote $\lambda_{max}:=\max_{i}\lambda_{i}$.

To determine the decision schedule in each time slot, similar to the scheme in \cite{nisri09}, each link
independently sends an INTENT message with probability $1/n$. If the transmitter of link $i$ sends the packet
and the packet is successfully received (indicated by an acknowledgement from the receiver), then link $i$ is
included in the decision schedule. So, the decision schedule which only includes link $i$ is chosen with
probability $\frac{1}{n}(1-\frac{1}{n})^{n-1}:=\frac{c_{n}}{n}$ where $c_{n}:=(1-\frac{1}{n})^{n-1}\ge
c_{min}:=0.2,\forall n$; and the decision schedule which includes no links is chosen with probability
$1-c_{n}$.

\begin{lemma}
\label{lem:mix_complete_graph} \emph{Let $\mu_{x',t}$ be the distribution of the transmission schedule in
slot $t$ with any initial schedule $x'$, we have}
\begin{equation}
||\mu_{x',t}-\pi||_{var}\le\gamma^{t}\label{eq:mixing},
\end{equation}
\emph{where} $\gamma=1-\frac{c_{min}}{n\cdot(1+\lambda_{max})}.$
\end{lemma}
\begin{proof}
Consider two copies of the Markov chain $X(t)$ and $Y(t)$. We will construct a coupling $\{X(t),Y(t)\}$, with
$Y(0)$ chosen from the stationary distribution $\pi$. With this coupling, we show that for any $x,y$,
\begin{eqnarray}
& & P(X(t+1)=Y(t+1) | X(t)=x, Y(t)=y) \nonumber \\
& \ge & \bar{c}:=c_{min}/[n\cdot(1+\lambda_{max})].\label{eq:one-step}
\end{eqnarray}

Since there is at most one link active in a given slot, for convenience, we write the state as $i$ if link
$i\in\{1,2,\dots,n\}$ is active in the slot, and 0 if no link is active.

If $x=y$, then the two Markov chains have already coupled in slot $t$, so that $P(X(t+1)=Y(t+1) | X(t)=x,
Y(t)=y)=1$. If $x \ne y$, there are three cases.

(i) $x=i_{1}$ and $y=i_{2}$ where $i_{1}\ne i_{2}$ and $i_{1},i_{2}\ne0$.

According to the Glauber dynamics defined above, with probability $c_{n}/n$, one link is \emph{selected} in
the \emph{decision schedule}, and with probability $1-c_{n}$ no link is selected. We define the following
\emph{coupling}.

If link $i_{1}$ is selected in the process $X(\cdot)$, then select link $i_{2}$ in the process $Y(\cdot)$.
WLOG, assume that $\lambda_{i_{1}}\ge\lambda_{i_{2}}$. Then, turn off link $i_{1}$ in $X(\cdot)$ w. p.
$1/(1+\lambda_{i_{1}})$. If link $i_{1}$ is turned off, then also turn off link $i_{2}$ in $Y(\cdot)$. If
link $i_{1}$ is not turned off, turn off link $i_{2}$ w. p.
$[1/(1+\lambda_{i_{2}})-1/(1+\lambda_{i_{1}})]/[\lambda_{i_{1}}/(1+\lambda_{i_{1}})]$. Then it is easy to see
that link $i_{2}$ is turned off w. p. $1/(1+\lambda_{i_{2}})$.

If link $i_{2}$ is selected in $X(\cdot)$, then select link $i_{1}$ in $Y(\cdot)$. Clearly, after the
selection, they cannot be turned on.

If a link other than $i_{1}$ and $i_{2}$ is selected in $X(\cdot)$, then select the same link in $Y(\cdot)$.
Also, any link selected must remain off at time 1.

If no link is selected in $X(\cdot)$, then also select no link in $Y(\cdot)$.

Therefore, \begin{eqnarray*}
& & P(X(t+1)=Y(t+1) | X(t)=x, Y(t)=y) \\
 & \ge & P(X(t+1)=Y(t+1)=0 | X(t)=x, Y(t)=y)\\
 & \ge & c_{min}/[n\cdot(1+\lambda_{i_{1}})]\\
 & \ge & c_{min}/[n\cdot(1+\lambda_{max})] = \bar{c}.\end{eqnarray*}

(ii) $x=i$ and $y=0$ where $i\ne0$.

In $X(\cdot)$ and $Y(\cdot)$, we choose the same link (in the decision schedule), and use the same coin toss
to decide whether the chosen link should try to turn on or off. Therefore, w. p. $c_{n}/n$, link $i$ is
chosen, and then $X(t+1)=Y(t+1)$. So,
\begin{eqnarray*}
& & P(X(t+1)=Y(t+1) | X(t)=x, Y(t)=y) \\
& \ge & c_{n}/n\ge c_{min}/n\ge \bar{c}.
\end{eqnarray*}

(iii) $x=0$ and $y=i$ where $i\ne0$.

This is symmetric to case (ii).

Therefore, in any case, (\ref{eq:one-step}) holds, and we have
$$P(X(t+1) \ne Y(t+1) | X(t)=x, Y(t)=y)\le 1-\bar{c},\forall x,y,t.$$
This implies that
\[ P(X(t)\ne Y(t)|X(0)=x')\le (1-\bar{c})^{t},\forall x'.\]

By the well-known coupling lemma (e.g., \cite{bubdye97}),
$$||\mu_{x',t}-\pi||_{var}\le P(X(t)\ne Y(t) | X(0)=x')$$
for any initial state $x',$ hence proving (\ref{eq:mixing}).
\end{proof}

\begin{theorem}
\emph{Given any $\rho<1$ (which is independent of $n$), PGD-CSMA can support ${\boldsymbol
\nu}\in\rho\Lambda$ with a mixing time of $O(n)$.}
\end{theorem}
\begin{proof}
For the complete interference graph,
$$\Lambda=\{{\boldsymbol
\nu}|\sum_{j=1}^{n}\nu_{j}\le1,\nu_{j}\ge0,\forall j\}.$$ Therefore, since ${\boldsymbol \nu}\in\rho\Lambda$,
we have $\sum_{j=1}^{n}\nu_{j}\le\rho$. So, ${\boldsymbol\nu}' := {\boldsymbol \nu}+\frac{1-\rho}{2n} {\bf
1}$ satisfies that $\sum_{j=1}^{n}\nu'_{j}\le \frac{1+\rho}{2}<1$. As a result, ${\boldsymbol \nu}' \in
\frac{1+\rho}{2}\Lambda \subset \Lambda^{o}$.

Let ${\boldsymbol \lambda}$ be the vector of fugacities such that $s_{i}$ under PGD-CSMA is equal to
 $\nu'_{i} > \nu_{i}$ at each link $i$. We have $s_{i}=\frac{\lambda_{i}}{1+\sum_{j=1}^{n}\lambda_{j}}=\nu'_{i},\forall i$ and\[
\sum_{j=1}^{n}s_{j}=\frac{\sum_{j=1}^{n}\lambda_{j}}{1+\sum_{j=1}^{n}\lambda_{j}}=\sum_{j=1}^{n}\nu'_{j} \le
\frac{1+\rho}{2}.\]

Therefore, $\lambda_{max}\le\sum_{j=1}^{n}\lambda_{j}\le\frac{1+\rho}{1-\rho}$. Using Lemma
\ref{lem:mix_complete_graph} and (\ref{eq:mixingbound}), we know that the mixing time \[
T_{mix}\le\frac{1}{1-\gamma}=\frac{n\cdot(1+\lambda_{max})}{c_{min}}\le\frac{2n}{c_{min}(1-\rho)}=O(n).\]
\end{proof}

\section{Conclusion}  \label{sec:conclusion}

In this paper, we have shown that Glauber dynamics based CSMA can result in queue lengths that only grow
polynomially in the network size in bounded degree graphs if the arrival rates lie within a certain fraction
of the capacity region. This establishes a positive result in contrast to previous results which showed that
polynomial queue lengths are not possible if the interference graph of the network can be arbitrary
\cite{shatsetsi09}. To establish our results, we use Markov chain coupling theory to estimate the mixing time
of parallel Glauber dynamics. It is interesting that our upper bound on the queue lengths is larger (in an
order sense) for the dynamic-parameter algorithm compared to the fixed-parameter algorithm. However, it is
unclear whether this is an artifact of our bounding techniques or if it is an inherent penalty due to the
time required for adaptation. It would be interesting to explore this issue in future work. We are also
interested to study  the throughput and delay performance of CSMA scheduling algorithms in other important
classes of interference graphs, and design enhanced or new algorithms to further improve the performance.

\appendix

\subsection{\label{conditions}Other conditions for fast mixing of parallel Glauber dynamics}

\begin{corollary} \label{coro:c1}
\emph{Let $m=\min_{v\in V}\frac{1+\lambda_v}{q_v}$, $M=\max_{v\in V}\frac{1+\lambda_v}{q_v},$ and
$\xi=\frac{M}{m}.$ If}
\begin{eqnarray}
\theta \triangleq \min_{v\in V}\left\{1+\lambda_v-\sum_{w\in\mathcal N_v}\lambda_w\right\} &
> & 0,
\end{eqnarray}
\emph{then we have}
\begin{eqnarray}
T_{\textit{mix}} & \leq & \bar{T}_{mix}  = \Big \lceil \frac{M}{\theta}\log(n\xi e) \Big\rceil.
\end{eqnarray}
\end{corollary}
\begin{proof}
Choose $f(v)=\frac{1+\lambda_v}{q_v}$, $\forall v\in V$ in Theorem \ref{thm:main}.
\end{proof}

\begin{corollary} \label{coro:c2}
\emph{Let $q_{min}=\min_{v\in V}q_v$, $q_{max}=\max_{v\in V}q_v,$ and $\xi=\frac{q_{max}}{q_{min}}.$ If}
\begin{eqnarray}
b \triangleq \max_{v\in V}\sum_{w\in\mathcal N_v}\frac{\lambda_w}{1+\lambda_w} &<& 1,
\end{eqnarray}
\emph{then we have}
\begin{eqnarray} \label{eq:mixingPGB2}
T_{\textit{mix}} & \leq & \bar{T}_{\textit{mix}} = \Big\lceil \frac{\log\big(n\xi e\big)}{q_{min}(1-b)}
\Big\rceil.
\end{eqnarray}
\end{corollary}
\begin{proof}
Choose $f(v)=\frac{1}{q_v}$, $\forall v\in V$ in Theorem \ref{thm:main}.
\end{proof}

\begin{remark}
\emph{Note that the condition $\lambda_v<\frac{1}{d_v-1}$ for all $v\in V$ in Corollary~\ref{coro:c3} might
be very different from $b=\max_{v\in V}\sum_{w\in\mathcal N_v}\frac{\lambda_w}{1+\lambda_w}<1$ in
Corollary~\ref{coro:c2}, e.g., in graphs of star topology.}
\end{remark}

\subsection{\label{proof:MRF}Proof of (\ref{eq:MRF})}

Let $\mathcal{A}:=\{\sigma\in\Omega|\sigma_{i}=1,\sigma_{j}=0,\forall j\in\mathcal{N}_{i}\}$. Then
$s_{i}=\sum_{\sigma\in\mathcal{A}}\pi(\sigma)$. Let
$\mathcal{B}:=\{\sigma\in\Omega|\sigma_{i}=0,\sigma_{j}=0,\forall j\in\mathcal{N}_{i}\}$. Clearly,
$\mathcal{A}\cap\mathcal{B}=\emptyset$. Define $\mathcal{C}=\{\sigma+{\bf e}_{i}|\sigma\in\mathcal{B}\}$,
where ${\bf e}_{i}$ is the $n$-dimensional vector whose $i$-th element is 1 and all other elements are 0. We
claim that $\mathcal{A}=\mathcal{C}$. (Indeed, any $\sigma'\in\mathcal{A}$ can be written as
$\sigma'=\sigma+{\bf e}_{i}$ for some $\sigma\in\mathcal{B}$. So $\sigma'\in\mathcal{C}$. Also, any
$\sigma'\in\mathcal{C}$ is in $\mathcal{A}$.) Therefore,
$\sum_{\sigma'\in\mathcal{A}}\pi(\sigma')=\sum_{\sigma\in\mathcal{B}}\pi(\sigma+{\bf e}_{i})$. By
(\ref{eq:productformPGD}), we have $\pi(\sigma+{\bf e}_{i})=\lambda_{i}\pi(\sigma)$. So
$s_{i}=\sum_{\sigma'\in\mathcal{A}}\pi(\sigma')=\lambda_{i}\sum_{\sigma\in\mathcal{B}}\pi(\sigma)$. As a
result, $p_{i,0}=\sum_{\sigma:\sigma_{j}=0,\forall
j\in\mathcal{N}_{i}}\pi(\sigma)=\sum_{\sigma'\in\mathcal{A}}\pi(\sigma')+\sum_{\sigma\in\mathcal{B}}\pi(\sigma)=(1+\frac{1}{\lambda_{i}})s_{i}$,
proving (\ref{eq:MRF}).

\end{document}